\documentclass[runningheads]{llncs}

\providecommand{\longv}[1]{#1}
\providecommand{\shortv}[1]{}

\usepackage[norelnum]{figures}
\usepackage[norelnum,lncs]{theorem_environments}

\usepackage{wrapfig}
\usepackage[utf8]{inputenc}

\usepackage[export]{adjustbox}
\usepackage[2cell,matrix,all,rotate]{xy}
\usepackage{amsmath}
\usepackage{amssymb}
\usepackage{mathtools}
\usepackage{csquotes}
\usepackage{fontawesome}

\spnewtheorem{fact}[theorem]{Fact}{\bfseries}{\itshape}

\newcommand{\pd}[2]{$#1 <_{C} #2$}
\newcommand{\pid}[2]{$#1 \nless_{C} #2$}
\newcommand{\sd}[2]{$#1 <_{D} #2$}
\newcommand{\sid}[2]{$#1 \nless_{D} #2$}
\newcommand{\crule}[1]{\mathit{check}_{#1}}

\usepackage{booktabs}
\usepackage{enumitem}

\usepackage[pdftex]{hyperref}

\title{Graph Consistency as a Graduated Property}
\subtitle{Consistency-Sustaining and -Improving Graph Transformations\footnote{%
	This is a corrected version of the 2020 ICGT paper. 
	Besides minor corrections, we omit the statement that direct sustainment implies sustainment from Theorem~\ref{thm:relations-sustainment}. 
	This statement is only true with additional preconditions. 
	For a correct treatment of that issue we refer to the extended journal version of this paper~\cite{KSTZ22}. 
	There, we slightly adapt the definition of sustainment such that direct sustainment always implies sustainment.}
}
\titlerunning{Graph Consistency as a Graduated Property}

\author{Jens Kosiol\inst{1}\,\textsuperscript{\faEnvelopeO}\orcidID{0000-0003-4733-2777} \and Daniel Str{\"u}ber\inst{2}\orcidID{0000-0002-5969-3521} \and Gabriele Taentzer\inst{1}\orcidID{0000-0002-3975-5238} \and Steffen Zschaler\inst{3}\orcidID{0000-0001-9062-6637}}
\authorrunning{J. Kosiol, D. Str{\"u}ber, G. Taentzer, and S. Zschaler}

\institute{%
  Philipps-Universität Marburg, Marburg, Germany\\
	\email{\{kosiolje,taentzer\}@mathematik.uni-marburg.de} \and
  Radboud University, Nijmegen, the Netherlands \email{d.strueber@cs.ru.nl} \and
  King's College London, London, UK \email{szschaler@acm.org}%
}

\begin{document}

  \maketitle

  \begin{abstract}
    Where graphs are used for modelling and specifying systems, consistency is an important concern. To be a valid model of a system, the graph structure must satisfy a number of constraints. To date, consistency has primarily been viewed as a binary property: a graph either is or is not consistent with respect to a set of graph constraints. This has enabled the definition of notions such as constraint-preserving and constraint-guaranteeing graph transformations.
    Many practical app\-li\-ca\-tions---for example model repair or evolutionary search---implicitly assume a more graduated notion of consistency, but without an explicit formalisation only limited analysis of these applications is possible.
    In this paper, we introduce an explicit notion of consistency as a graduated property, depending on the number of constraint violations in a graph. We present two new characterisations of transformations (and transformation rules) enabling reasoning about the gradual introduction of consistency: while consistency-sustaining transformations do not decrease the consistency level, consistency-improving transformations strictly reduce the number of constraint violations. We show how these new definitions refine the existing concepts of constraint-preserving and constraint-guaranteeing transformations. To support a static analysis based on our characterisations, we present criteria for deciding which form of consistency ensuring transformations is induced by the application of a transformation rule. We illustrate our contributions in the context of an example from search-based model engineering.
  \end{abstract}

  \section{Introduction}
 
Graphs and graph transformations~\cite{EEPT06} are a good means for system modelling and specification. Graph structures naturally relate to the structures typically found in many (computer) systems and graph transformations provide intuitive tools to specify the semantics of a model or implement refinement and analysis techniques for specifications. 
  
  In all of these scenarios, it is important that the graphs used are consistent; that is, that their structures satisfy a set of constraints. Some constraints can be captured by typing graphs over so-called type graphs~\cite{EEPT06}---these allow capturing basic structural constraints such as which kinds of nodes may be connected to each other. 
  To allow the expression of further constraints, the theory of nested graph constraints has been introduced~\cite{HP09}. A graph is considered consistent if it is correctly typed and satisfies all given constraints. Note that this notion of consistency is binary: a graph either is consistent or it is not consistent. It is impossible to distinguish different degrees of consistency. 
  
  In software engineering practice, it is often necessary to live with, and manage, a degree of inconsistency~\cite{Nuseibeh+01}. This requires tools and techniques for identifying, measuring, and correcting inconsistencies. In the field of graph-based specifications, this has led to many practical applications, where a more fine-grained notion of graph consistency is implicitly applied. For example, research in model repair has aimed to automatically produce graph-transformation rules that will gradually improve the consistency of a given graph. Such a rule may not make a graph completely consistent in one transformation step, but performing a sequence of such transformations will eventually produce a consistent graph (\emph{e.g.,}~\cite{HS18,NKR17,NRA17,SH19}). In the area of search-based model engineering (\emph{e.g.,}~\cite{BZJ19,Fleck+15}),
  rules are required to be applicable to inconsistent graphs and, at least, not to produce new inconsistencies. In earlier work, we have shown how such rules can be generated at least with regard to multiplicity constraints~\cite{BZJ19}. However, in all of these works, the notion of \enquote{partial} graph consistency remains implicit. Without explicitly formalising this notion, it becomes difficult to reason about the validity of the rules generated or the correctness of the algorithm by which these rules were produced. 
  
  In this paper, we introduce a new notion of graph consistency as a graduated property. A graph can be consistent \emph{to a degree,} depending on the number of constraint violations that occur in the graph. 
  This conceptualisation allows us to introduce two new characterisations of graph transformations: a \emph{con\-sis\-ten\-cy-sus\-tai\-ning} transformation does not decrease the overall consistency level, while a \emph{con\-sis\-ten\-cy-im\-pro\-ving} transformation strictly decreases the number of violations in a graph. We lift these characterisations to the level of graph transformation rules, allowing rules to be characterised as consistency sustaining and consistency improving, respectively. We show how these definitions fit with the already established terminology of constraint-preserving and constraint-guaranteeing transformations / rules. 
  Finally, we introduce formal criteria that allow checking whether a given graph-transformation rule is consistency sustaining or consistency improving w.r.t. constraints in specific forms. 
  
  Thus, the contributions of our paper are:
  \begin{enumerate}
    \item We present the first formalisation of graph consistency as a graduated property of graphs;
    \item We present two novel characterisations of graph transformations and transformation rules with regard to this new definition of graph consistency and show how these refine the existing terminology;
    \item We present static analysis techniques for checking whether a graph-trans\-for\-ma\-tion rule is consistency sustaining or improving. 
  \end{enumerate}
  
  The remainder of this paper is structured as follows: We introduce a running example in Sect.~\ref{sec:example} before outlining some foundation terminology in Sect.~\ref{section:preliminaries}. Section~\ref{sec:sustaining-and-improving} introduces our new concepts and Sect.~\ref{sec:criteria} discusses how graph-trans\-for\-ma\-tion rules can be statically analysed for these properties. A discussion of related work in Sect.~\ref{section:related_work} concludes the paper. 
	\shortv{All proofs are provided in an extended version of this paper ~\cite{KSTZ20EV}.}
	\longv{The proofs of all results in this paper can be found in Appendix~\ref{sec:appendix}.}

  \section{Example}
\label{sec:example}

  Consider \textit{class responsibility assignment} (CRA, \cite{bowman2010solving}), a standard problem in ob\-ject-oriented software analysis.
  Given is a set of features (methods, fields) with dependencies between them.
 The goal is to create a set of classes and assign the features to classes so that a certain \textit{fitness function} is maximized.
The fitness function rewards the assignment of dependent features to the same class (cohesion), while punishing dependencies that run between classes (coupling) and solutions with too few classes.
 Solutions can be expressed as instances of the type graph shown in the left of Fig.~\ref{fig:mutationrules}.
 For realistic problem instances, an exhaustive enumeration of all  solutions to find the optimal one is not feasible.
	
  Recently, a number of works have addressed the CRA problem via a combination of graph transformation and meta-heuristic search techniques, specifically evolutionary algorithms \cite{fleck2016class,struber2017generating,BZJ19}.
  An evolutionary algorithm uses genetic operators such as cross-over and mutation to find optimal solution candidates in an efficient way.
  In this paper, we focus on mutation operators, which have been specified using graph transformation rules in these works.

  \begin{figure}[tb]
    \centering
      \includegraphics[width=0.15\textwidth,valign=t,draft=false]{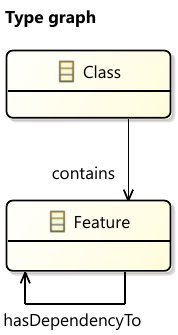}
      \includegraphics[width=0.83\textwidth,valign=t,draft=false]{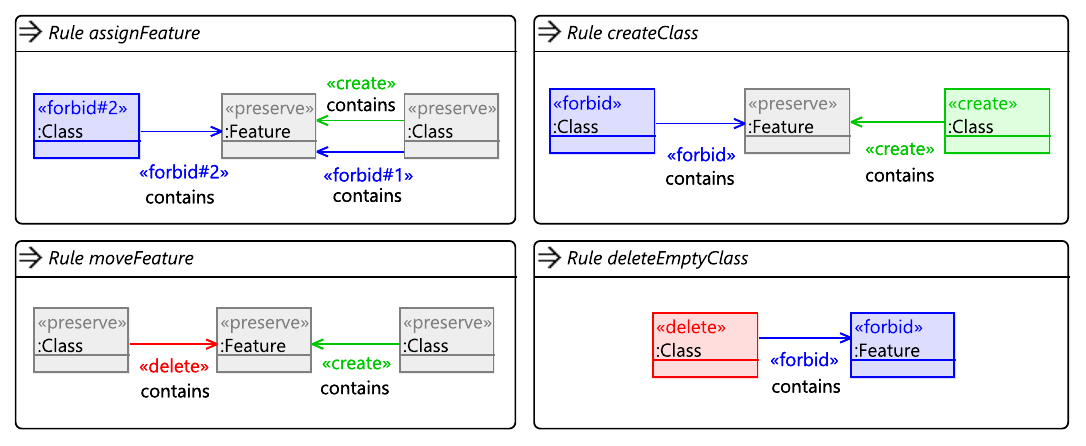}
    \caption{Type graph and four mutation rules for the CRA problem.}
    \label{fig:mutationrules}
  \end{figure}

  Figure~\ref{fig:mutationrules} depicts four mutation rules for the CRA problem, taken from the available MDEOptimiser solution \cite{BZS18}.
	The rules are specified as graph transformation rules \cite{EEPT06} in the Henshin notation~\cite{ABJKT10,struber2017henshin}:
  Rule elements are tagged as \textit{delete}, \textit{create}, \textit{preserve} or \textit{forbid},  which denotes them as being included in the LHS, the RHS, in both rule sides, or a NAC. 
  Rule \textit{assignFeature} assigns a randomly selected as-yet-unassigned feature to a class.
  Rule \textit{createClass} creates a class and assigns an as-yet-unassigned feature to it.
  Rule \textit{moveFeature} moves a feature between two classes.
  Rule \textit{deleteEmptyClass} deletes a class to which no feature is assigned.

  Solutions in an optimization problem such as the given one usually need to be consistent with regard to the constraints given by the problem domain.
  We consider three constraints for the CRA case:
  \begin{enumerate}[label=($c_\arabic*$), itemsep=0pt]
  \item Every feature is contained in at most one class.
  \item Every class contains at least one feature.
  \item If a feature $F_1$ has a dependency to another feature $F_2$, \\ and $F_2$ is contained in a different class than $F_1$, \\ then $F_1$ must have a dependency to a feature $F_3$ in the same class.
  \end{enumerate}

  Constraints $c_1$ and $c_2$ come from Fleck et al.'s formulation of the CRA problem \cite{fleck2016class}.
  Constraint $c_3$ can be considered a \textit{helper constraint} (compare \textit{helper objectives }\cite{jensen2004helper}) that aims to enhance the efficiency of the search by formulating a constraint with a positive impact to the fitness function:
  Assigning dependent features to the same class is likely to improve coherence.

  Given an arbitrary solution model (valid or invalid), mutations may introduce new violations.
	For example, applying \textit{moveFeature} can leave behind an empty class, thus violating $c_2$.
  While constraint violations can potentially be removed using  repair techniques \cite{NRA17,HS18,SH19}, these can be computationally expensive and may involve strategies that lead to certain regions of the search space being preferred, threatening the efficiency of the search.
  Instead, it would be desirable to design mutation operators that impact consistency in a positive or at least neutral way.
  Each application of a mutation rule should contribute to some particular violations being removed, or at least ensure that the degree of consistency does not decrease.
  Currently, there exists no formal framework for identifying such rules.
  The established notions of constraint-preserving and constraint-guaranteeing rules \cite{HP09} assume an already-valid model or a transformation that removes all violations at once; both are infeasible in our scenario.

  \section{Preliminaries}
\label{section:preliminaries}
  
  Our new contributions are based on typed graph transformation systems following the double-pushout approach~\cite{EEPT06}. 
  We implicitly assume that all graphs, also the ones occurring in rules and constraints, are typed over a common type graph $\mathit{TG}$; that is, there is a class $\mathit{Graph}_{TG}$ of graphs typed over $\mathit{TG}$. 
  A \emph{nested graph constraint}~\cite{HP09} is a tree of injective graph morphisms. 
    
    \begin{define}[caption={(Nested) graph conditions and constraints}]
      Given a graph $P$, a \emph{(nested) graph condition} over $P$ is defined recursively as follows:
      \texttt{true} is a graph condition over $P$ and if $a: P \hookrightarrow C$ is an injective morphism and $d$ is a graph condition over $C$, $\exists \, (a: P \hookrightarrow C, d)$ is a graph condition over $P$ again. 
      If $d_1$ and $d_2$ are graph conditions over $P$, $\neg d_1$ and $d_1 \wedge d_2$ are graph conditions over $P$.
      A \emph{(nested) graph constraint} is a condition over the empty graph $\emptyset$. 
      
      A condition or constraint is called \emph{linear} if the symbol $\wedge$ does not occur, i.e., if it is a (possibly empty) chain of morphisms.
      The \emph{nesting level} $\mathit{nl}$ of a condition $c$ is recursively defined by setting $\mathit{nl}(\texttt{true}) \coloneqq 0$, $\mathit{nl}(\exists \, (a: P \hookrightarrow C, d))  \coloneqq \mathit{nl}(d) + 1$, $\mathit{nl}(\neg d) \coloneqq \mathit{nl}(d)$, and $\mathit{nl}(d_1\wedge d_2) \coloneqq \max (\mathit{nl}(d_1), \mathit{nl}(d_2))$. 
      Given a graph condition $c$ over $P$, an injective morphism $p: P \hookrightarrow G$ \emph{satisfies} $c$, written $p \models c$, if the following applies:
      Every morphism satisfies \texttt{true}. 
      The morphism $p$ satisfies a condition of the form $c = \exists \, (a: P \hookrightarrow C, d)$ if there exists an injective morphism $q: C \hookrightarrow G$ such that $p = q \circ a$ and $q$ satisfies $d$. 
      For Boolean operators, satisfaction is defined as usual. 
      A graph $G$ satisfies a graph constraint $c$, denoted as $G \models c$, if the empty morphism to $G$ does so. 
      A graph constraint $c_1$ \emph{implies} a graph constraint $c_2$, denoted as $c_1 \Rightarrow c_2$, if $G \models c_1 \Rightarrow G \models c_2$ for all graphs $G$. 
      The constraints are \emph{equivalent}, denoted as $c_1 \equiv c_2$, if $c_1 \Rightarrow c_2$ and $c_2 \Rightarrow c_1$.
    \end{define}
   
    In the notation of graph constraints, we drop the domains of the involved morphisms and occurrences of \texttt{true} whenever they can unambiguously be inferred. 
    For example, we write $\exists (C,\neg \exists C^\prime)$ instead of $\exists (\emptyset \hookrightarrow C, \neg \exists (a: C \hookrightarrow C^\prime, \texttt{true}))$. 
    Moreover, we introduce $\forall (C, d)$ as an abbreviation for the graph constraint $\neg \exists (C, \neg d)$. 
    Further sentential connectives like $\vee$ or $\Rightarrow$ can be introduced as abbreviations as usual (which is irrelevant for linear constraints).

    We define a normal form for graph conditions that requires that the occurring quantifiers alternate. 
    For every linear condition there is an equivalent condition in this normal form~\cite[Fact~2]{SH19}.
    
    \begin{define}[caption={Alternating quantifier normal form (ANF)}]
      A linear condition $c$ with $\mathit{nl}(c) \geq 1$ is in \emph{alternating quantifier normal form} (ANF) when the occurring quantifiers alternate, i.e., if $c$ is of the form $Q(a_1,\Bar{Q}(a_2, Q(a_3, \dots)\dots)$ with $Q \in \{\exists,\forall\}$ and $\Bar{\exists} = \forall, \Bar{\forall}= \exists$, none of the occurring morphisms $a_i$ is an isomorphism, and the only negation, if any, occurs at the innermost nesting level (i.e., the constraint is allowed to end with \texttt{false}). If a constraint in ANF starts with $\exists$, it is called \emph{existential}, otherwise it is called \emph{universal}.
    \end{define}
    
    \begin{lemma}[Non-equivalence of constraints in ANF]\label{lem:non-equivalence-anf}
      Let $c_1 = \exists(C_1,d_1)$ and $c_2 = \forall(C_2,d_2)$ be constraints in ANF. 
      Then $c_1 \not\equiv c_2$.
    \end{lemma}
    
We have $c_1 \not\equiv c_2$ since $\emptyset \models c_2$ but $\emptyset \not\models c_1$.
    Lemma~\ref{lem:non-equivalence-anf} implies that the first quantifier occurring in the ANF of a constraint separates linear constraints into two disjoint classes. 
    This ensures that our definitions in Section~\ref{sec:sustaining-and-improving} are meaningful. 
    
     \medskip
    \emph{Graph transformation} is the rule-based modification of graphs. The following definition recalls graph transformation as a double-pushout.
    
    \begin{define}[caption={Rule and transformation}]
    A {\em plain rule} $r$ is defined by $p = (L  \hookleftarrow K\hookrightarrow R)$ with $L, K,$ and $R$ being graphs connected by two graph inclusions. An application condition $ac$ for $p$ is a condition over $L$. A {\em rule} $r = (p,ac)$ consists of a plain rule $p$ and an application condition $ac$ over $L$.
    
	A {\em transformation (step)} $G \Rightarrow_{r,m} H$ which applies rule $r$ to a graph $G$ consists of two pushouts as depicted in Fig.~\ref{fig:rule-application}. 
	Rule $r$ is {\em applicable} at the injective morphism $m: L \rightarrow G$ called {\em match} if $m \models ac$ and there exists a graph $D$ such that the left square is a pushout.
	Morphism $n$ is called \emph{co-match}.
	Morphisms $g$ and $h$ are called {\em transformation morphisms}. 
    The {\em track morphism}~\cite{Plump05} of a transformation step $G \Rightarrow_{r,m} H$ is the partial morphism $tr: G \dashrightarrow H$ defined by $tr(x) = h(g^{-1}(x))$ for $x \in g(D)$ and undefined otherwise.
    \end{define}
    
    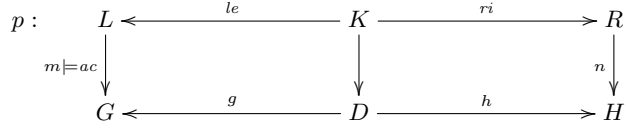
\begin{figure}[tbp]
        \begin{center}
        $$
          \xymatrix@C=16pt@R=8pt{
            p: & L \ar@{<-}[rrrr]^{\mathit{le}}\ar@{->}[dd]_{m \models \mathit{ac}} & & & & K \ar@{->}[rrrr]^{\mathit{ri}}\ar@{->}[dd] & & & & R\ar@{->}[dd]_{n} \\
               &                                      & & & &                                  & & & & \\
               & G \ar@{<-}[rrrr]^{g}                 & & & & D \ar@{->}[rrrr]^{h}            & & & & H \\
          }
        $$
        \end{center}
         \vspace{-\baselineskip}
        \caption{Rule application}
        \label{fig:rule-application}
				\vspace{-.5\baselineskip}
    \end{figure}
    
    Obviously, transformations interact with the validity of graph constraints. 
    Two well-studied notions are constraint-guaranteeing and -preserving transformations~\cite{HP09}. 
    \begin{define}[caption={$c$-guaranteeing and -preserving transformation}]
      Given a constraint $c$, a transformation $G \Rightarrow_{r,m} H$ 
      is \emph{$c$-guaranteeing} if $H\models c$.  
      Such a transformation is \emph{$c$-preserving} if $G \models c \Rightarrow H \models c$. 
      A rule $r$ is \emph{$c$-guaranteeing} (\emph{$c$-preserving}) if every transformation via $r$ is.
    \end{define}
    
    As we will present criteria for consistency sustainment and improvement based on conflicts and dependencies of rules, we recall these notions here as well. 
    Intuitively, a transformation step \emph{causes a conflict} on another one if it hinders this second one. 
    A transformation step is \emph{dependent} on another one if it is first enabled by that. 
    
    \begin{define}[caption={Conflict}]
	Let a pair of transformations $(t_1, t_2): (G \Rightarrow_{m_1,r_1} H_1, G \Rightarrow_{m_2,r_2} H_2)$ applying rules $r_i = (L_i \hookleftarrow K_i \hookrightarrow R_i, ac_i),\ i = 1,2$ be given such that $t_i$ yields transformation morphisms $G \stackrel{g_i}{\leftarrow} D_i \stackrel{h_i}{\rightarrow} H_i$. 
	\emph{Transformation pair}  $(t_1,t_2)$ is \emph{conflicting} (or \emph{$t_1$ causes a conflict on $t_2$}) if there does not exist a morphism $x: L_2 \rightarrow D_1$ such that $g_1 \circ x = m_2$ and $h_1 \circ x \models \mathit{ac}_2$. 
	\emph{Rule pair} $(r_1,r_2)$ is \emph{conflicting} if there exists a conflicting transformation pair $(G \Rightarrow_{m_1,r_1} H_1, G \Rightarrow_{m_2,r_2} H_2)$. 
	If $(r_1,r_2)$ and $(r_2,r_1)$ are both not conflicting, rule pair $(r_1,r_2)$ is called {\em parallel independent}.
 \end{define}
 
 \begin{define}[caption={Dependency}]
	Let a sequence $t_1; t_2: G \Rightarrow_{m_1,r_1} H_1 \Rightarrow_{m_2,r_2} X$ of transformations applying rules $r_i =  (L_i \hookleftarrow K_i \hookrightarrow R_i, ac_i),\ i= 1,2$ be given such that $t_1$ yields transformation morphisms $G \stackrel{g_1}{\leftarrow} D_1 \stackrel{h_1}{\rightarrow} H_1$. 
	Transformation  $t_2$ is \emph{dependent on} $t_1$ if there does not exist a morphism $x: L_2 \rightarrow D_1$ such that $h_1 \circ x = m_2$ and $g_1 \circ x \models ac_2$. 
	Rule $r_2$ is {\em dependent on} rule $r_1$ if there exists a transformation sequence $t_1;t_2: G \Rightarrow_{m_1,r_1} H_1 \Rightarrow_{m_2,r_2} X$ such that $t_2$ is dependent on $t_1$.
	If $r_1$ is not dependent on $r_2$ and $r_2$ is not dependent on $r_1$, rule pair $(r_1,r_2)$ is called {\em sequentially independent}. 
	
	A \emph{weak critical sequence} is a sequence $t_1; t_2: G \Rightarrow_{m_1,r_1} H_1 \Rightarrow_{m_2,r_2} X$ of transformations such that $t_2$ depends on $t_1$, $n_1$ and $m_2$ are jointly surjective (where $n_1$ is the co-match of $t_1$), and $m_i$ is \emph{not} required to satisfy $\mathit{ac}_i$ ($i=1,2$).
 \end{define}
  As rule $r_2$ in a rule pair $(r_1,r_2)$ will always be plain in this paper, a transformation step can cause a conflict on another one if and only if it deletes an element that the second transformation step matches.
 Similarly, a transformation step can depend on another one if and only if the first step creates an element that the second matches or deletes an edge that is adjacent to a node the second one deletes.

  \section{Consistency-sustaining and consistency-im\-pro\-ving rules and transformations}
\label{sec:sustaining-and-improving}

  In this section, we introduce our key new concepts. We do so in three stages, first introducing foundational definitions for partial consistency, followed by a generic definition of consistency sustainment and improvement. Finally, we give stronger definitions for which we will be able to provide a static analysis in Sect.~\ref{sec:criteria}.
    
  \subsection{Partial consistency}

  To support the discussion and analysis of rules and transformations that improve graph consistency, but do not produce a fully consistent graph in one step, we introduce the notion of \emph{partial consistency.} 
   We base this notion on relating the number of \emph{constraint violations} to the total number of \emph{relevant occurrences} of a constraint. 
   For the satisfaction of an existential constraint, a single valid occurrence is enough. 
   In contrast, universal constraints require the satisfaction of some sub-constraint for every occurrence. 
   Hence, the resulting notion is binary in the existential case, but graduated in the universal one.
   
   \medskip
   
   \noindent\fbox{%
    \parbox{.96\linewidth}{%
        In the remainder of this paper, a \emph{constraint} is always a linear constraint in ANF having a nesting level $\geq 1$.\footnotemark\ Moreover, all graphs are finite.
        }%
    }
    \footnotetext{Requiring nesting level $\geq 1$ is no real restriction as constraints with nesting level $0$ are Boolean combinations of \texttt{true} which means they are equivalent to \texttt{true} or \texttt{false}, anyhow. In contrast, restricting to linear constraints actually excludes some interesting cases. We believe that the extension of our definitions and results to also include the non-linear case will be doable. Restricting to the linear case first, however, makes the statements much more accessible and succinct.} 
   
   \begin{define}[caption={Occurrences and violations}]
     Let $c = Q(\emptyset \to C,d)$ with $Q \in \{\exists, \forall\}$ be a constraint. 
     An \emph{occurrence} of $c$ in a graph $G$ is an injective morphism $p: C \hookrightarrow G$, and $\mathit{occ}(G,c)$ denotes the \emph{number of such occurrences}.
     
     If $c$ is universal, its number of \emph{relevant  occurrences} in a graph $G$, denoted as $\mathit{ro}(G,c)$, is defined as $\mathit{ro}(G,c) \coloneqq \mathit{occ}(G,c)$ and its \emph{number of constraint violations}, denoted as $\mathit{ncv}(G,c)$, is the number of occurrences $p$ for which $p \not\models d$.
      
     If $c$ is existential, $\mathit{ro}(G,c) \coloneqq 1$ and $\mathit{ncv}(G,c) \coloneqq 0$ if there exists an occurrence $p: C \hookrightarrow G$ such that $p \models d$ but $\mathit{ncv}(G,c) \coloneqq 1$ otherwise. 
   \end{define}
   
   \begin{define}[caption={Partial consistency}]
     Given a graph $G$ and a constraint $c$, $G$ is \emph{consistent} w.r.t. $c$ if $G \models c$.
     The \emph{consistency index of $G$ w.r.t. $c$} is defined as
     \begin{equation*}
       \mathit{ci}(G,c) \coloneqq 1 - \frac{\mathit{ncv}(G,c)}{\mathit{ro}(G,c)}
     \end{equation*}
     where we set $\frac{0}{0} \coloneqq 0$.  
     We say that $G$ is \emph{partially consistent} w.r.t. $c$ if $\mathit{ci}(G,c) > 0$. 
   \end{define}
   
   The next proposition makes precise that the
   consistency index runs between $0$ and $1$ and indicates the degree of consistency a graph $G$ has w.r.t. a constraint~$c$. 
   
	\begin{fact}[Consistency index]\label{fact:consistency-index}
    Given a graph $G$ and a constraint $c$, then $0 \leq ci(G,c) \leq 1$ and $G \models c$ if and only if $ci(G,c) = 1$. 
    Consistency implies partial consistency. 
    Moreover, $ci(G,c) \in \{0,1\}$ for an existential constraint. 
	\end{fact}

			\begin{figure}[t]
			\centering
				\includegraphics[width=1.00\textwidth,draft=false]{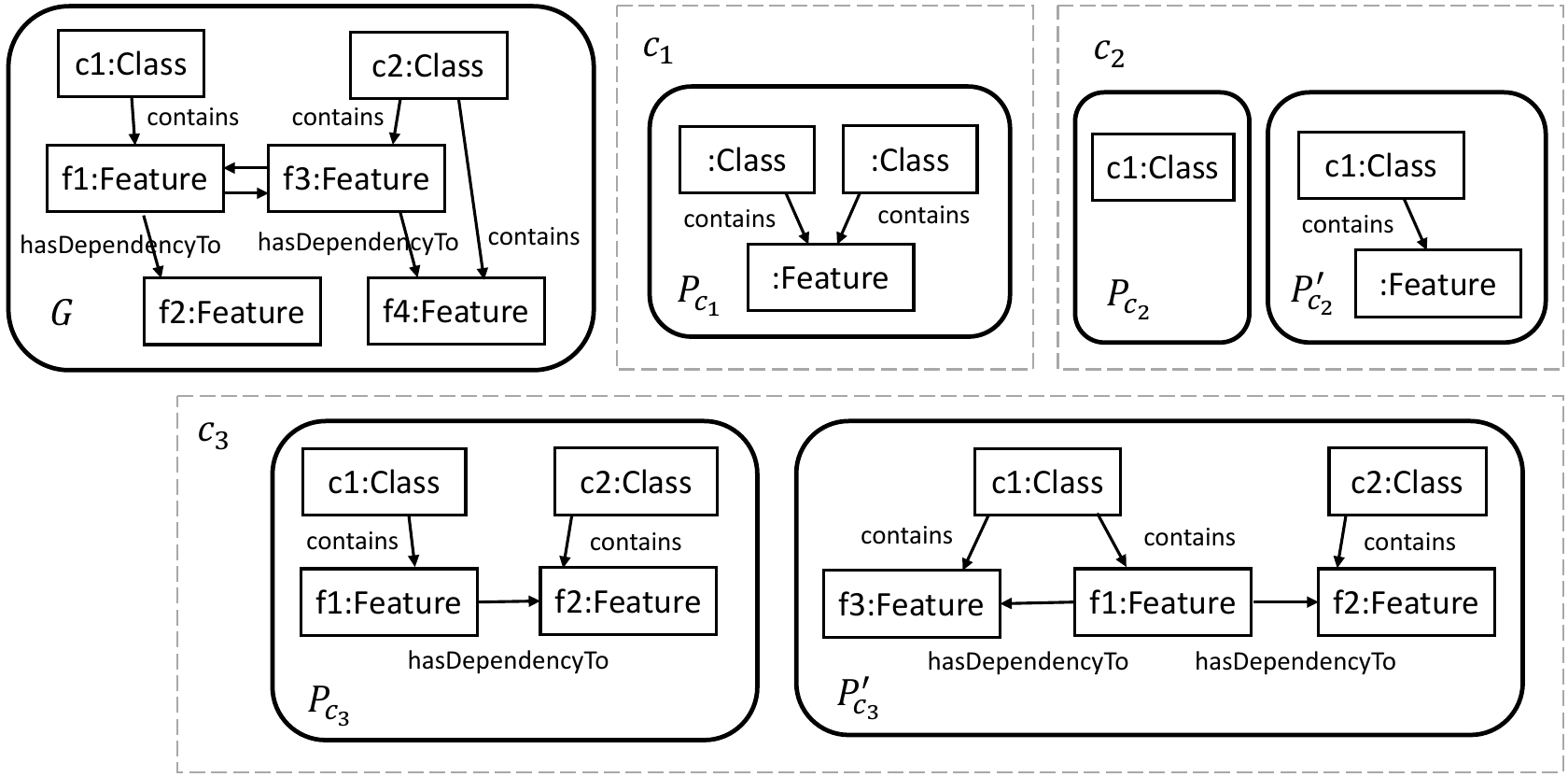}
			\caption{Example constraints and graph.}
			\label{fig:example-graphbased}
		\end{figure}

	\begin{example}
Based on Fig.~\ref{fig:example-graphbased}, we can express the three informal constraints from Section~\ref{sec:example} as nested graph constraints.
Constraint $c_1$ can be expressed as $\neg \exists P_{c_1}$, constraint $c_2$ becomes $\forall (P_{c_2}, \exists P_{c_2}^\prime)$,  
and constraint $c_3$ becomes $\forall (P_{c_3}, \exists P_{c_3}^\prime)$.
Graph $G$ (in the left top corner of Fig.~\ref{fig:example-graphbased}) satisfies $c_1$ and $c_2$. 
It does not satisfy $c_3$, since we  cannot find an occurrence of $P_{c_3}^\prime$ for the occurrence of  $P_{c_3}$ in $G$ where \textit{f1} and \textit{f2} are mapped to \textit{f1} and \textit{f3}, respectively. 
	  Graph $G$ in Fig.~\ref{fig:example-graphbased} has the consistency index 0.5 with regard to~$c_3$, since one violation exists, and two non-violating occurrences are required. 
  \end{example}

  \subsection{Consistency sustainment and improvement}
  
  In the remainder of this section, our goal is to introduce the notions of \emph{con\-sis\-ten\-cy-sus\-tai\-ning} and \emph{consistency-improving rule applications} which refine the established notions of preserving and guaranteeing applications~\cite{HP09}. 
  
  \begin{define}[caption={Consistency sustainment and improvement}]\label{def:non_direct_sustainment_and_improvement}
    Given a graph constraint $c$ and a rule $r$, a transformation $t: G \Rightarrow_{r,m} H$ is \emph{consistency sustaining} w.r.t. $c$ if $\mathit{ci}(G,c) \leq \mathit{ci}(H,c)$. 
    It is \emph{consistency improving} if it is consistency sustaining, $\mathit{ncv}(G,c) > 0$, and $\mathit{ncv}(G,c) > \mathit{ncv}(H,c)$. 
    
    The rule $r$ is \emph{consistency sustaining} if all of its applications are. 
    It is \emph{consistency improving} if all of its applications are consistency sustaining and there exists a graph $G \in \mathit{Graph}_{TG}$ with $\mathit{ncv}(G,c) > 0$ and a consistency-improving transformation $G \Rightarrow_{r,m} H$. 
    A consistency improving rule is \emph{strongly consistency improving} if all of its applications to graphs $G$ with $\mathit{ncv}(G,c) > 0$ are consistency-improving transformations. 
  \end{define}
  In the above definition, we use the number of constraint violations (and not the consistency index) to define improvement to avoid an undesirable side-effect: 
  Defining improvement via a growing consistency index would lead to consistency-improving transformations (w.r.t. a universal constraint) which do not repair existing violations but only create new valid occurrences of the constraint. 
  Hence, there would exist infinitely long transformation sequences where every step increases the consistency index but validity is never restored. 
  Consistency-im\-prov\-ing transformations, and therefore \emph{strongly} consistency improving rules, require that the number of constraint violations strictly decreases in each step. Therefore, using only such transformations and rules, we cannot construct infinite transformation sequences.
  
  Any consistency-improving rule can be turned into a strongly consistency-improving rule if suitable pre-conditions can be added that restrict the applicability of the rule only to those cases where it can actually improve a constraint violation. This links the two forms of consistency-improving rules to their practical applications: in model repair~\cite{NKR17,SH19} we want to use rules that will only make a change to a graph when there is a violation to be repaired---strongly consistency-improving rules. However, in evolutionary search~\cite{BZJ19}, we want to allow rules to be able to make changes even when there is no need for repair, but to fix violations when they occur; consistency-improving rules are well-suited here as they can be applied even when no constraint violations need fixing.

  \subsection{Direct consistency sustainment and improvement}

  While the above definitions are easy to state and understand, it turns out that they are inherently difficult to investigate. 
  Comparing numbers of (relevant) occurrences and violations allows for very disparate behavior of consistency-sustaining (-improving) transformations: 
  For example, a transformation is allowed to destroy as many valid occurrences as it repairs violations and is still considered to be consistency sustaining w.r.t. a universal constraint.
  
  Next, we introduce further qualified notions of consistency sustainment and improvement. 
  The idea behind this refinement is to \emph{retain} the validity of occurrences of a universal constraint: valid occurrences that are preserved by a transformation are to remain valid. 
	In this way, sustainment and improvement become more \emph{direct} as it is no longer possible to compensate for introduced violations by introducing additional valid occurrences. 
  The notions of (direct) sustainment and improvement are related to one another and also to the already known ones that preserve and guarantee constraints. 
  In Sect.~\ref{sec:criteria} we will show how these stricter definitions allow for static analysis techniques to identify consistency-sustaining and -improving rules.  
  
  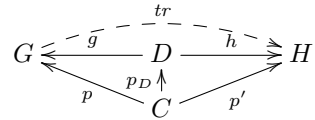
\begin{wrapfigure}{r}{0.34\textwidth} 
        \centering
				\vspace{-20pt}
          \xymatrix@C=16pt@R=8pt{
                G \ar@{<-}[rr]^{g}\ar@{-->}@/^1pc/[rrrr]^{\mathit{tr}}                  &  & D \ar@{->}[rr]^{h}             &  &  H \\
                                                     &  &  C \ar@{->}[llu]^{p}\ar@{->}[rru]_{p^{\prime}}\ar@{->}[u]^{p_D}  & &\\
          }
					\vspace{-5pt}
        \caption{Rule application with morphisms from a graph $C$, occurring in some constraint}
        \label{fig:interaction-rule-application-occurrence}
    \end{wrapfigure}

		The following definitions assume a transformation step to be given and relate occurrences of constraints in its start and result graph as depicted in Fig.~\ref{fig:interaction-rule-application-occurrence}. 
  The existence of a morphism $p_D$ such that the left triangle commutes (and $p^\prime$ might be defined as $h \circ p_D$) is equivalent to the tracking morphism $\mathit{tr}: G \dashrightarrow H$ being a total morphism when restricted to $p(C)$ which is equivalent to the transformation not destroying the occurrence $p$. 
    
    \begin{define}[caption={Direct consistency sustainment}]
      Given a graph constraint $c$, a transformation $t: G \Rightarrow_{m,r} H$ via rule $r$ at match $m$ with trace $tr$ (Fig.~\ref{fig:interaction-rule-application-occurrence}) is \emph{directly consistency sustaining} w.r.t. $c$ if either $c$ is existential and the transformation is $c$-preserving or $c = \forall (C,d)$ is universal and
      \begin{align*}
        & \forall p: C \hookrightarrow G \big( (p \models d \wedge \mathit{tr} \circ p \text{ is total}) \Rightarrow \mathit{tr} \circ p \models d\big) \wedge \\
        & \forall p^\prime: C \hookrightarrow H \big(\neg \exists p: C \hookrightarrow G \left( p^\prime = \mathit{tr} \circ p \right) \Rightarrow p^\prime \models d\big) \enspace .
      \end{align*}
       
     A {\em rule} $r$ is \emph{directly consistency sustaining} w.r.t. $c$ if all its applications are. 
    \end{define}
		
    \begin{wraptable}{r}{0.5\linewidth}
      \vspace{-10pt}
      \centering
      \caption{Properties of example rules.}
      \label{table:exampleProperties}	
      \resizebox{0.5\textwidth}{!}{%
        \setlength{\tabcolsep}{0.2em}
        \begin{tabular}{r|ccc|ccc}
          \toprule
                            & \multicolumn{3}{c}{\textbf{Consistency}} & \multicolumn{3}{|c}{\textbf{Consistency}} \\ 
                            & \multicolumn{3}{c}{\textbf{sustaining }} & \multicolumn{3}{|c}{\textbf{improving}} \\ 
          \textbf{Rule}                  &  \ \textbf{$c_1$} \ &  \ \textbf{$c_2$}   & \textbf{$c_3$} \  & \ \textbf{$c_1$} \  & \ \textbf{$c_2$} \  & \textbf{$c_3$} \  \\
          \hline
          \textit{assignFeature}     & +    & +    & - & -    & +    & - \\
          \textit{createClass}      & +    & +    & - & -    & -    & - \\
          \textit{moveFeature}       & (+)    & -    & - & -    & -    & - \\ 
          \textit{deleteEmptyClass}  & +    & +    & + & -    & +* & - \\
          \bottomrule
								\multicolumn{7}{l}{ \hspace{30pt}  \textbf{Legend}:    + denotes \textit{directly}, (+) denotes} \\
		\multicolumn{7}{l}{  \hspace{70pt} \textit{non-directly}, * denotes \textit{strongly} }\\

        \end{tabular}
      }
      \vspace{-20pt}
    \end{wraptable}    
    The first requirement in the definition checks that constraints that were already valid in $G$ are still valid in $H$, unless their occurrence has been removed; that is, the transformation must not make existing valid occurrences invalid. Note, however, that we do not require that the constraint be satisfied by the same extension, just that there is still a way to satisfy the constraint at that occurrence. The second requirement in the definition checks that every ``new'' occurrence of the constraint in $H$ satisfies the constraint; that is, the transformation must not introduce fresh violations. 

  The following theorem relates the new notions of (direct) consistency sustainment to preservation and guarantee of constraints.
  
  \begin{theorem}[Sustainment relations]\label{thm:relations-sustainment}
    Given a graph constraint $c$, every $c$-gua\-ran\-tee\-ing transformation is directly consistency-sustaining 
		and every con\-sis\-ten\-cy-sustaining transformation is $c$-preserving. 
    The analogous implications hold on the rule level: 
    
		\vspace{-\baselineskip}
    {
    \small
    $$
      \xymatrix@C=16pt@R=8pt{
            \mbox{constraint-preserving rule} \ar@{<=}[rrrr]^{\mbox{\cite{HP09}}} \ar@{<=}[dd]^{\mathit{Thm.~\ref{thm:relations-sustainment}}} & & & & \mbox{constraint-guaranteeing rule} \ar@{=>}[dd]^{\mathit{Thm.~\ref{thm:relations-sustainment}}}\\
            & & & & \\
            \mbox{consistency-sustaining rule} & & & & \mbox{directly consistency-sustaining rule}
            }
    $$
    }
  \end{theorem}

  The following example illustrates these notions and shows that sustainment is different from constraint guaranteeing or preserving.

    \begin{example}\label{exm:sustaining}
	Table~\ref{table:exampleProperties} denotes for each rule from the running example if it is consistency sustaining w.r.t.\ each constraint.
      Rule \textit{createClass} is directly consistency sustaining w.r.t.\ $c_1$ (no double assignments) and $c_2$ (no empty classes), since it cannot assign an already assigned feature or remove existing assignments.
			However, it is not consistency guaranteeing, since it cannot remove any violation either.
			Rule \textit{moveFeature} is consistency sustaining w.r.t $c_1$, but not directly so, since it can introduce new violations, but only while at the same time removing another violation, leading to a neutral outcome.
			Starting with the plain version of rule \emph{createClass} and computing a preserving application condition for constraint $c_1$ according to the construction provided by Habel and Pennemann~\cite{HP09} results in the application condition depicted in Fig.~\ref{fig:generated-preserving-application-condition}.
			By construction, equipping the plain version of \emph{createClass} with that application condition results in a consistency-preserving rule.
			However, whenever applied to an invalid graph, the antecedent of this application condition evaluates to \texttt{false} and, hence, the whole application condition to \texttt{true}. 
			In particular, the rule with this application condition might introduce further violations of $c_1$ and is, thus, not sustaining.  
			
			\begin{figure}[tb]
			  \centering
		      \includegraphics[width=1.00\textwidth,draft=false]{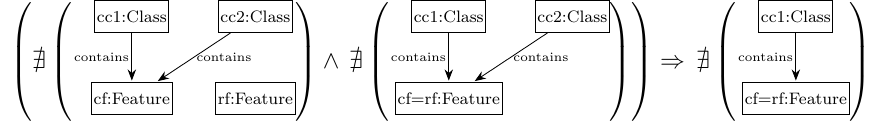}
			\caption{Generated preserving application condition for \emph{createClass} w.r.t. constraint $c_1$. The feature named \textsf{rf} is the one from the LHS of \emph{createClass}.}
			\label{fig:generated-preserving-application-condition}
		\end{figure}
  \end{example}
  
  Similarly, the \emph{direct} notion of consistency improvement preserves the validity of already valid occurrences in the case of universal constraints and degenerates to the known concept of constraint-guarantee in the existential case. 
    \begin{define}[caption={Direct consistency improvement}]
      Given a graph constraint $c$, a transformation $t: G \Rightarrow_{m,r} H$ via rule $r$ at match $m: L \hookrightarrow G$ with trace $tr$ (Fig.~\ref{fig:interaction-rule-application-occurrence}) is \emph{directly consistency improving} w.r.t. $c$ if $G \nvDash c$, the transformation is directly consistency sustaining, and either $c$ is existential and the transformation is $c$-guaranteeing or $c = \forall (C,d)$ is universal and
      \begin{align*}
        & \exists p: C \hookrightarrow G \big( p \nvDash d \wedge p^\prime \coloneqq \mathit{tr} \circ p \text{ is total}\, \wedge p^\prime \models d \big) \vee \\
        & \exists p: C \hookrightarrow G \big(p \nvDash d \wedge p^\prime \coloneqq \mathit{tr} \circ p \text{ is not total}\,\big)
      \end{align*}

      We lift the notion of directly consistency-improving transformations to the level of rules in the same way as in Def.~\ref{def:non_direct_sustainment_and_improvement}. 
			This leads to directly consistency-improving rules and a strong form of directly consistency-improving rules. 
    \end{define}

  (Direct) consistency improvement is related to, but different from constraint guarantee and consistency sustainment as made explicit in the next theorem.
  
  \begin{theorem}[Improvement relations]\label{thm:relations-improvement}
    Given a graph constraint $c$, every directly consistency-improving transformation is a consistency-improving transformation and every consistency-improving transformation is consistency sustaining. 
    Moreover, every $c$-guaranteeing transformation starting from a graph $G$ that is inconsistent w.r.t. $c$ is a directly consistency-improving transformation. 
    The analogous implications hold on the rule level, provided that there exists a match for the respective rule $r$ in a graph $G$ with $G \not\models c$:  
    
    \vspace{-\baselineskip}
    {
      \small
      $$
        \xymatrix@C=16pt@R=8pt{
           \mbox{consistency-sustaining rule}
           \ar@{<=}[dd]_{\mathit{Thm.~\ref{thm:relations-improvement}}}& & & & \mbox{$c$-guaranteeing rule} \ar@{=>}[llll]^{\mathit{Thm.~\ref{thm:relations-sustainment}}} \ar@{=>}[dd]_{\mathit{Thm.~\ref{thm:relations-improvement}}}\\
          & & & & \\
          \mbox{consistency-improving rule}  \ar@{<=}[rrrr]^{\mathit{Thm.~\ref{thm:relations-improvement}}} & & & & \mbox{directly consistency-improving rule}
        }
      $$
    }
  \end{theorem}
		
  \begin{example}\label{exm:preserving-sustaining-guaranteeing}
	        Table~\ref{table:exampleProperties} denotes for each rule of the running example if it is consistency improving w.r.t.\ each constraint.
				For example, the rule \textit{deleteEmptyClass} is directly strongly consistency improving but not guaranteeing w.r.t. $c_2$ (no empty classes), since it always removes a violation (empty class) but generally not all violations in one step.	
	Rule \textit{assignFeature} is directly consistency improving w.r.t. $c_2$ but not strongly so: 
	It can turn empty classes into non-empty ones but does not do so in every possible application.
	Rule \textit{createClass} is consistency sustaining but not improving w.r.t. $c_2$, as it cannot reduce the number of empty~classes.
  \end{example}

  \section{Static Analysis for Direct Consistency Sustainment and Improvement}
\label{sec:criteria}
    
    In this section, we consider specific kinds of constraints and present a static analysis technique for direct consistency sustainment and improvement. 
    We present criteria for rules to be \emph{directly} consistency sustaining or \emph{directly} consistency improving w.r.t. these kinds of constraint. 
    The restriction to specific kinds of constraint greatly simplifies the presentation; at the end of the section we hint at how our results may generalize to arbitrary universal constraints. 
    
    The general idea behind our static analysis technique is to check for validity of a constraint by applying a trivial (non-modifying) rule that just checks for the existence of a graph occurring in the constraint. 
    This allows us to present our analysis technique in the language of \emph{conflicts and dependencies} which has been developed to characterise the possible interactions between rule applications~\cite{Plump05,EEPT06}. 
    As a bonus, since the efficient detection of such conflicts and dependencies has been the focus of recent theoretical and practical research~\cite{LBKST19,LSTBH18}, we obtain tool support for an automated analysis based on Henshin.
   
   \medskip
   
   \noindent\fbox{%
    \parbox{.98\linewidth}{%
        In the remainder of this paper, we assume the following setting: Let $r = (L \hookleftarrow K \hookrightarrow R, ac)$ be a rule, $c$ a graph constraint of the form $\neg\exists C = \forall(\emptyset \hookrightarrow C, \texttt{false})$ and $d$ a graph constraint of the form $\forall(C, \exists C^{\prime}) = \forall(\emptyset \hookrightarrow C, \exists a: C \hookrightarrow C^{\prime})$. Given a graph $G$, there is the rule $\crule{G} \coloneqq G \xhookleftarrow{id_G} G \xhookrightarrow{id_G} G$ given. 
        }%
    }
  
  \medskip
  For the statement of the following results, note that sequential independence of the (non-modifying) rule $\crule{C}$ from $r$ means that $r$ cannot create a new match for $C$. 
	Similarly, parallel independence of $\crule{C^{\prime}}$ from $r$ means that $r$ cannot destroy a match for $C^{\prime}$.
  We first state criteria for direct consistency sustainment: 
	If a rule cannot create a new occurrence of $C$, it is directly consistency sustaining w.r.t. a constraint of the form $\neg\exists C$. 
	If, in addition, it cannot delete an occurrence of $C^\prime$, it is directly consistency sustaining w.r.t. a constraint of the form $\forall (C, \exists C^\prime)$. 
    
    \begin{theorem}[Criteria for direct consistency sustainment]\label{thm:criterion-consistency-sustainment}
      Rule $r$ is directly consistency sustaining w.r.t. constraint $c$ if and only if $\crule{C}$ is sequentially independent from $r$. 
      If, in addition, $\crule{C^\prime}$ is parallel independent from $r$, then $r$ is directly consistency sustaining w.r.t. constraint $d$. 
    \end{theorem}
    The above criterion is sufficient but not necessary for constraints of the form $\forall (C, \exists C^\prime)$. 
    For example, it does not take into account the possibility of $r$ creating a new valid occurrence of $C$. 
    The next proposition strengthens the above theorem by partially remedying this. 
    \begin{proposition}\label{prop:criterion-consistency-sustainment2}
      If $\crule{C^\prime}$ is parallel independent from $r$ and for every weak critical sequence $G \Rightarrow_{r,m} H \Rightarrow_{\crule{C},p^{\prime\prime}} H$ it holds that there is an injective morphism $q^{\prime\prime}: C^{\prime} \hookrightarrow H$ with $q^{\prime\prime} \circ a = p^{\prime\prime}$, i.e., $p^{\prime\prime} \models \exists C^{\prime}$, then $r$ is directly consistency sustaining w.r.t. constraint $d$. 
    \end{proposition}

    For consistency improvement we state criteria on rules as well: 
    If a rule is directly consistency improving w.r.t. a  constraint of the form $\forall (C, \exists C^{\prime})$, it is either (1) able to destroy an occurrence of $C$ (deleting a part of it) or (2) to bring about an occurrence of $C^{\prime}$ (creating a part of it). 
		In case (2), we can even be more precise: 
		The newly created elements do not stem from $C$ but from the part of $C^{\prime}$ without $C$; this is what the formula in the next theorem expresses. 
		For constraints of the form $\neg\exists C$, condition (1) is the only one that holds. 
    
    \begin{theorem}[Criteria for direct consistency improvement]\label{thm:criterion-consistency-improvement}
        If rule $r$ is directly consistency sustaining w.r.t. constraint $c$, then it is directly consistency improving w.r.t. $c$ if and only if $r$ causes a conflict for $\crule{C}$. 
        If $r$ is directly consistency improving w.r.t. constraint $d$, then $r$ causes a conflict for $\crule{C}$ or $\crule{C^{\prime}}$ is sequentially dependent on $r$ in such a way that 
        \begin{equation*}
          n(R \setminus K) \cap p^\prime (C^\prime) \subseteq p^\prime (C^\prime \setminus a(C))
        \end{equation*}
        where, in this dependency, $n$ is the co-match of the first transformation applying $r$ and $p^\prime$ is the match for $\crule{C^{\prime}}$.
    \end{theorem}
    
    The above criterion is not sufficient in case of constraint $d$. 
    The existing conflicts or dependencies do not ensure that actually an \emph{invalid} occurrence of $C$ can be deleted or a new occurrence of $C^\prime$ can be created in such a way that an invalid occurrence of $C$ is \enquote{repaired}. 

    Looking closer to the criteria stated above, we can find some recurring patterns. Table~\ref{tab:generalisation-criteria} lists the kinds of universal constraints up to nesting level 2 and the corresponding criteria. While we have shown the criteria in the first two rows in Theorems~\ref{thm:criterion-consistency-sustainment} and \ref{thm:criterion-consistency-improvement}, we conjecture the criteria in the last row of Table~\ref{tab:generalisation-criteria}. To prove  generalized theorems for nesting levels $\geq 2$, however, is up to future work. 
		
    \begin{table}[t]
        \centering
        \caption{Generalisation of the criteria from Theorems~\ref{thm:criterion-consistency-sustainment} and \ref{thm:criterion-consistency-improvement} to universal constraints up to nesting level 2. Here, $ck_C$ is short for $\mathit{check}_C$, \sd{r_1}{r_2} denotes dependency of $r_2$ on $r_1$, \pd{r_1}{r_2} denotes $r_1$ causing a conflict for $r_2$, and crossed out versions denote the respective absence.}
        \label{tab:generalisation-criteria}
        \resizebox{\textwidth}{!}{%
        \begin{tabular}{@{}l|l|l@{}}
            \toprule
            \textbf{type of constr.} & \textbf{crit. for direct consist. sust.} & \textbf{crit. for direct consist. impr.}\\
            \midrule
            $\forall(C, \texttt{false}) \equiv \neg\exists C$ & \sid{r}{\mathit{ck}_{C}} & \pd{r}{\mathit{ck}_{C}}\\
            $\forall (C_1, \exists C_2)$ & \sid{r}{\mathit{ck}_{C_1}} $\wedge$ \pid{r}{\mathit{ck}_{C_2}} & \pd{r}{\mathit{ck}_{C_1}} $\vee$ \sd{r}{\mathit{ck}_{C_2}} \\
            $\forall (C_1, \exists (C_2, \neg\exists C_3))$ & \sid{r}{\mathit{ck}_{C_1}} $\wedge$ \pid{r}{\mathit{ck}_{C_2}} $\wedge$ \sid{r}{\mathit{ck}_{C_3}} & \pd{r}{\mathit{ck}_{C_1}} $\vee$ \sd{r}{\mathit{ck}_{C_2}} $\vee$ \pd{r}{\mathit{ck}_{C_3}}\\
            \bottomrule
        \end{tabular}}
    \end{table} 

\begin{table}[t]
\centering
\caption{Applying the criteria from Tbl.~\ref{tab:generalisation-criteria} to the example; $ck_C$ is short for $\mathit{check}_C$.}
\label{table:exampleChecks}
\setlength{\tabcolsep}{0.2em}
\begin{tabular}{r|ccc|cc|ccc|cc}
\toprule
									& \multicolumn{5}{c}{\textbf{Consis. sust. (suff. cr.)}} & \multicolumn{5}{|c}{\textbf{Consis. impr. (necc. cr.)}} \\ 
									& \multicolumn{3}{c}{seq. indep.} & \multicolumn{2}{|c}{par. indep.} & \multicolumn{3}{|c}{par. dep.} & \multicolumn{2}{|c}{seq. dep.}  \\ 
\textbf{Rule}                  & \footnotesize{$ck_{P_{c_1}}$} & \footnotesize{$ck_{P_{c_2}}$}  & \footnotesize{$ck_{P_{c_3}}$} & \footnotesize{$ck_{P'_{c_2}}$}  & \footnotesize{$ck_{P'_{c_3}}$}  & \footnotesize{$ck_{P_{c_1}}$} & \footnotesize{$ck_{P_{c_2}}$}  & \footnotesize{$ck_{P_{c_3}}$} & \footnotesize{$ck_{P'_{c_2}}$}  & \footnotesize{$ck_{P'_{c_3}}$} \\
\hline
\small{\textit{assignFeature}}    & -    & +    & - & +    & +     
                                  & -    & -    & - & +    & +	 \\
\small{\textit{createClass}}      & -    & -    & - & +    & +     
                                  & -    & -    & - & +    & +   \\
\small{\textit{moveFeature}}      & -    & +    & - & -    & -   
                                  & +    & -    & + & +    & +   \\
\small{\textit{deleteEmptyClass}} & +    & +    & + & +    & +    
                                  & -    & +    & - & -    & -    \\
\bottomrule
\end{tabular}
\end{table}

\begin{example}
We can use the criteria in Table \ref{tab:generalisation-criteria} to semi-automatically reason about consistency sustainment and improvement in our example.
To this end, we first apply automated conflict and dependency analysis (CDA, \cite{LSTBH18}) to the relevant pairs of mutation and check rules.
Using the detected conflicts and dependencies, we infer parallel and sequential (in)dependence per definition, as shown in Table~\ref{table:exampleChecks}.
For example, since no dependencies between \textit{assignFeature} and $check_{P_{c_1}}$ exist, we conclude that these rules are sequentially independent. 

\textit{Consistency sustainment}: Based on Table~\ref{table:exampleChecks}, we find that the sufficient criterion formulated in Theorem~\ref{thm:criterion-consistency-sustainment} is adequate to show direct consistency sustainment in four out of seven positive cases as per Table~\ref{table:exampleProperties}:
rule \textit{assignFeature} with constraint $c_2$ and rule \textit{deleteEmptyClass} with constraints $c_1$, $c_2$ and $c_3$. 
Moreover, the stronger criterion in Proposition~\ref{prop:criterion-consistency-sustainment2} allows to recognize the case of \textit{createClass} with $c_2$. 
Discerning the remaining two positive cases (\textit{assignFeature} with $c_1$; \textit{createClass} with $c_1$) from the five negative ones requires further inspection. 

\textit{Consistency improvement}: Based on Table~\ref{table:exampleChecks}, our necessary criterion allows to detect the two positive cases in Table~\ref{table:exampleProperties}:
rules \textit{deleteEmptyClass} and \textit{assignFeature} with constraint $c_2$.
The former is due to parallel dependence, the latter due to sequential dependence (where inspection of the CDA results reveals a critical sequence with a suitable co-match).
The criterion is also fulfilled in six negative cases:  \textit{assignFeature} with $c_3$, \textit{createClass} with $c_2$ and $c_3$, and \textit{moveFeature} with $c_1$, $c_2$ and $c_3$. 
Four negative cases are correctly ruled out by the criterion.
\end{example}

  \section{Related Work}
\label{section:related_work}
    In this paper, we introduce a graduated version of a specific logic on graphs, namely of nested graph constraints. 
    Moreover, we focus on the interaction of this graduation with graph transformations. 
    Therefore, we leave a comparison with fuzzy or multi-valued logics (on graphs) to future work.
    Instead, we focus on works that also investigate the interaction between the validity of nested graph constraints and the application of transformation rules. 
    
    Given a graph transformation (sequence) $G \Rightarrow H$, the validity of graph $H$ can be established with basically three strategies: (1) graph $G$ is already valid and this validity is preserved, (2) graph $G$ is not valid and there is a $c$-guaranteeing rule applied, and (3) graph $G$ is made valid by a graph transformation (sequence) step-by-step.
    
    Strategies (1) and (2) are supported by the incorporation of constraints in application conditions of rules as presented in \cite{HP09} for nested graph constraints in general and implemented in Henshin~\cite{NKAT19}. 
    As the applicability of rules enhanced in that way can be severely restricted, improved constructions have been considered of specific forms of constraints. 
    For constraints of the form $\forall (C, \exists C\rq{})$, for example, a suitable rule scheme is constructed in~\cite{KFNST19}. 
    In~\cite{BLDBG11} refactoring rules are checked for the preservation of constraints of nesting level $ \leq 2$. 
		In~\cite{NKAT19}, two of the present authors also suggested certain simplifications of application conditions; the resulting ones are still constraint-preserving. 
		In~\cite{NKAT20}, we even showed that they result in the logically weakest application condition that is still directly consistency sustaining. 
		However, the result is only shown for negative constraints of nesting level one. 
		A very similar construction of negative application conditions from such negative constraints has very recently been suggested in~\cite{BSH20}.
    
    Strategy (3) is followed in most of the rule-based graph repair or model repair approaches. In~\cite{NRA17}, the violation of mainly multiplicity constraints is considered. In~\cite{HS18}, Habel and Sandmann derive graph programs from graph constraints of nesting level $\leq 2$. In~\cite{SH19}, they extend their results to constraints in ANF which end with $\exists C$ or constraints of one of the forms $\exists(C, \neg\exists C')$ or $\neg \exists C$. They also  investigate whether a given set of rules allows to repair such a given constraint.
    In~\cite{DG17} Dyck and Giese present an approach to automatically check whether a transformation sequence yields a graph that is valid w.r.t. specific constraints of nesting level $\leq 2$.
    
    Up to now, result graphs of transformations have been considered either valid or invalid w.r.t. to a graph constraint; intermediate consistency grades have not been made explicit. Thereby, $c$-preserving and $c$-guaranteeing transformations~\cite{HP09} focus on the full validity of the result graphs. 
    Our newly developed notions of consistency-sustainment and improvement are located properly in between existing kinds of transformations (as proven in Theorems~\ref{thm:relations-sustainment} and \ref{thm:relations-improvement}). 
		These new forms of transformations make the gradual improvements in consistency explicit. 
		While a detailed and systematic investigation (applying the static methods developed in this paper) is future work, a first check of the kinds of rules generated and used in \cite{KTRK16} (model editing), \cite{NRA17} (model repair), and \cite{BZJ19} (search-based model engineering) reveals that---in each case---at least some of them are indeed (directly) consistency-sustaining. 
    We are therefore confident that the current paper formalizes properties of rules that are practically relevant in diverse application contexts. 
		Work on partial graphs as in, e.g. \cite{SemerathVarro17}, investigates the validity of constraints in families of graphs which is not our focus here and therefore, not further considered. 
    
    Stevens in~\cite{Stevens14} discusses similar challenges in the specific context of bidirectional transformations. Here, consistency is a property of a pair of models (loosely, graphs) rather than between a graph and constraint. In this sense, it may be argued that our formalisation generalises that of~\cite{Stevens14}. Several concepts are introduced that initially seem to make sense only in the specific context of bidirectional transformations (\emph{e.g.,} the idea of $\stackrel{\rightarrow}{R}$ candidates), but may provide inspiration for a further extension of our framework with corresponding concepts. 

  \section{Conclusions}

  In this paper, we have introduced a definition of graph consistency as a graduated property, which allows for graphs to be partially consistent w.r.t. a nested graph constraint, inducing a partial ordering between graphs based on the number of constraint violations they contain. Two new forms of transformation can be identified as consistency sustaining and consistency improving, respectively. They are properly located in between the existing notions of constraint-preserving and constraint-guaranteeing transformations. Lifting them to rules, we have presented criteria for determining whether a rule is consistency sustaining or improving w.r.t. a graph constraint.  We have demonstrated how these criteria can be applied in the context of a case study from search-based model engineering.
  
  While the propositions we present allow us to check a given rule against a graph constraint, their lifting to a set of constraints is the next step to go. Furthermore, algorithms for constructing consistency-sustaining or -improving rules from a set of constraints are left for future work.
  
  \paragraph{Acknowledgements.}
  
    We thank the ICGT reviewers for their insightful and helpful comments. This work has been partially supported by DFG grants TA 294/17-1 and 413074939.
    
  \bibliographystyle{splncs04}
  \bibliography{biblio}
  
	\longv{%
		\appendix
		\newpage
		\section{Detailed Proofs}
\label{sec:appendix}
  In this appendix we present the detailed proofs of all statements of the paper. 
  
  \begin{proof}[of Lemma~\ref{lem:non-equivalence-anf}]
    For trivial reasons, $\emptyset \models c_2$ and $\emptyset \not\models c_1$:
    Since no morphism occurring in $c_1$ or $c_2$ is an isomorphism, $C_1 \neq \emptyset \neq C_2$. 
    Hence, there does not exist a morphism from $C_i$ to $\emptyset$ for $i=1,2$. \qed
  \end{proof}
  
  \begin{proof}[of Fact~\ref{fact:consistency-index}]
    First, $0 \leq \mathit{ci}(G,c) \leq 1$ since in any case $\mathit{ncv}(G,c) \leq \mathit{ro}(G,c)$, i.e., $0 \leq \frac{\mathit{ncv}(G,c)}{\mathit{ro}(G,c)} \leq 1$. 
    
    Moreover, $\mathit{ci}(G,c) = 1$ if and only if $\mathit{ncv}(G,c) = 0$ if and only if $G \models c$ for all $c \in \mathcal{C}$.
    
    The last claim for existential constraints follows from the fact that $\frac{\mathit{ncv}(G,c)}{\mathit{ro}(G,c)} \in \{0,1\}$ by definition of $\mathit{ncv}(G,c)$ and $\mathit{ro}(G,c)$. \qed
  \end{proof}

  \begin{proof}[of Theorem~\ref{thm:relations-sustainment}]
    Throughout the proof, let $c$ be the relevant constraint and $G \Rightarrow H$ a transformation.
    
    We first show that a $c$-guaranteeing transformation is directly consistency sustaining.
    By definition, guarantee of a constraint implies its preservation~\cite{HP09}. 
    In particular, the statement that guarantee implies direct sustainment is true in the case of existential constraints. 
    For the universal case, by $H \models c$, either $\mathit{occ}(G,c) = \mathit{ro}(G,c) = 0$ or $\mathit{ncv}(G,c) = 0$.
    In either case, the definition of direct consistency sustainment is met.
    
    %
    
    Finally, let a consistency sustaining transformation be given. 
    If already \\ $\mathit{ci}(G,c) = 1$, then this implies $\mathit{ci}(H,c) = 1$.
    In particular, the transformation is $c$-preserving. 
    
    Since the above statements are true on the transformation level, they can be directly lifted to the rule level. \qed
  \end{proof}
  
  \begin{proof}[of Theorem~\ref{thm:relations-improvement}]
    Again, throughout the proof, let $c$ be the relevant constraint and $G \Rightarrow H$ a transformation. 
    Note that, for both notions of improvement, by definition an improving transformation $G \Rightarrow H$ assumes $G \not\models c$.
    
    Therefore, first, let $G \Rightarrow H$ be a $c$-guaranteeing transformation where $G \not\models c$. 
    By Theorem~\ref{thm:relations-sustainment}, this transformation is consistency sustaining, in particular. 
    Hence, in case $c$ is an existential constraint, the transformation is directly consistency improving by definition.
    Therefore, let $c$ be a universal constraint.
    $G \not\models c$ implies that there is a morphism $p: C \hookrightarrow G$ with $p \not \models d$. 
    As $H \models c$ by definition of $c$-guaranteeing rule applications, either $\mathit{tr} \circ p$ is not total or $\mathit{tr} \circ p \models d$. 
    This means, either the first or the second condition of the definition of a directly consistency-improving transformation is met and therefore the transformation is directly consistency improving.
    
    In the following we show that every directly consistency-improving transformation is consistency improving. 
    The last claim, that every consistency-improving transformation is consistency sustaining, again holds by definition.
    
    First, every directly consistency-improving transformation is directly consistency sustaining by definition and by Theorem~\ref{thm:relations-sustainment} every directly consistency-sustaining transformation is consistency sustaining. 
    This means, we only have to check the conditions on the number of constraint violations.
    By assumption $\mathit{ncv}(G,c) > 0$.
    In case $c$ is an existential constraint, $H \models c$: 
    In that case, by definition, the transformation is even $c$-guaranteeing. 
    Hence, 
    \begin{equation*}
      \mathit{ncv}(G,c) = 1 > 0 = \mathit{ncv}(H,c)
    \end{equation*}
    and the transformation is consistency-improving.
    In case $c$ is universal, there exists (at least) one occurrence $p: C \hookrightarrow G$ that meets either the first or the second condition of the formula.
    In either case, this has the effect of decreasing $\mathit{ncv}(G,c)$ by one. 
    Moreover, direct consistency sustainment ensures that no new occurrences that violates the constraint is introduced. 
    In summary, $\mathit{ncv}(G,c) > \mathit{ncv}(H,c)$ and the transformation is consistency improving.
    
    On the rule level, (direct) consistency improvement is defined in such a way that at least one (directly) consistency-improving transformation via that rule needs to exist. 
    Hence, the proven statements on the transformation level lift to the rule level as long as there exists a $c$-guaranteeing transformation starting at an inconsistent graph $G$ via that rule. \qed 
  \end{proof}
  
  We formulate a technical lemma that we are going to use in the proof of the next theorems. 
  It relates the track morphism of a transformation to occurrences of constraints. 
  \begin{lemma}\label{lem:trace-morphism}
    Given a transformation $G \Rightarrow H$ and an occurrence $p: C \hookrightarrow G$ of a constraint $c$ in $G$, the track morphism $\mathit{tr}: G \dashrightarrow H$ is total, when restricted to $p(C)$ (i.e., $\mathit{tr} \circ p$ is a total morphism) if and only if there exists a morphism $p_D: C \hookrightarrow D$ such that $p = g \circ p_D$.
  \end{lemma}
  
  \begin{proof}
    For the first direction, set $p_D(x) \coloneqq g^{-1}(p(x))$ for all $x \in C$. 
    Since $g^{-1}(p(C))$ belongs to the domain of $\mathit{tr}$ by assumption, this results in a graph morphism with the desired property. 
    
    In the other direction, also the existence of $p_D: C \hookrightarrow D$ with $p = g \circ p_D$ states that $p(C)$ belongs to the domain of $\mathit{tr}$, i.e., $\mathit{tr} \circ p$ is total. \qed
  \end{proof}
  
  \begin{proof}[of Theorem~\ref{thm:criterion-consistency-sustainment}]
    We first consider the case of constraints $c = \forall (C,\texttt{false}) = \neg\exists C$.
    Assume that rule $r$ is directly consistency sustaining. 
    Let there be any transformation $t: G \Rightarrow_{r,m} H$ such that there is an injective morphism $p^\prime: C \hookrightarrow H$ (i.e., a match for $\crule{C}$). 
    Since $p^\prime \not\models \texttt{false}$, the second condition on directly consistency-sustaining transformations implies that there exists a morphism $p_D: C \hookrightarrow D$ such that $p^\prime = h \circ p_D$ (see Lemma~\ref{lem:trace-morphism}). 
    This means that the application of $\crule{C}$ is sequentially independent from $t$; and since $t$ was arbitrary, $\crule{C}$ is sequentially independent from $r$. 
    
    Conversely, assume $\crule{C}$ to be sequentially independent from $r$.
    Let there be a transformation $t: G \Rightarrow_{r,m} H$ and an injective morphism $p^{\prime}: C \hookrightarrow H$. 
    This morphism can be understood as a match for $\crule{C}$.  
    By the definition of sequential independence, there is an injective morphism $p_D: C \hookrightarrow D$ such that $p^{\prime} = h \circ p_D$. 
    This implies $p^{\prime} = \mathit{tr} \circ p$ (see Lemma~\ref{lem:trace-morphism}) where $p \coloneqq g \circ p_D$ and $\mathit{tr}$ is the trace morphism corresponding to the transformation. 
    This implies that both conditions in the definition of direct consistency improvement quantify over the empty set in that case; hence, rule $r$ is directly consistency improving.
    
    Secondly, we consider the case of constraint $d$. 
    Assume rule $r$ to be such that $\crule{C}$ is sequentially independent from $r$ and $r$ does not cause a conflict for $\crule{C^{\prime}}$. 
    As in the case of constraint $c$, the sequential independence implies that the second condition in the definition of direct consistency sustainment quantifies over an empty set in this case. 
    Hence, it is trivially true.
    We use the parallel independence to show that also the first condition is met.
  
    For this, let $G \Rightarrow_{r,m} H$ be a transformation step from $G$ to $H$ via rule $r$ at match $m$ where $D$ is the context graph of that transformation step. 
    Let $p: C \hookrightarrow G$ be a valid occurrence of $c$ such that there exists a morphism $p_D: C \hookrightarrow D$ with $p = g \circ p_D$ (compare Figure~\ref{fig:interaction-rule-application-occurrence}). 
    By validity of the occurrence, there exists an injective morphism $q: C^{\prime} \to G$ such that $q \circ a = p$.
    We have to show that there exists an injective morphism $q^{\prime}: C^{\prime} \hookrightarrow H$ such that $q^{\prime} \circ a = p^{\prime}$ where $p^{\prime} \coloneqq h \circ p_D$.

    The morphism $q$ can be understood as a match for $\crule{C^{\prime}}$ in $G$ and since $r$ does not cause a conflict for $\crule{C^{\prime}}$, there is an injective morphism $q_D: C^{\prime} \hookrightarrow D$ such that $q = g \circ q_D$. 
    One first computes
    \begin{align*}
        g \circ q_D \circ a &= q \circ a \\
                            &= p \\
                            &= g \circ p_D
    \end{align*}
    which implies $q_D \circ a = p_D$ since $g$ is injective.
    This can then be used to compute
    \begin{align*}
        q^{\prime} \circ a  &= h \circ q_D \circ a \\
                        &= h \circ p_D \\
                        &= p^{\prime} 
    \end{align*}
    as desired. \qed
  \end{proof}
  
  \begin{proof}[of Proposition~\ref{prop:criterion-consistency-sustainment2}]
    For any transformation $G \Rightarrow_{r,m} H$ that does not create a new occurrence of $C$ the argument is exactly the same as in the above proof of Theorem~\ref{thm:criterion-consistency-sustainment}. 
    Also, by absence of conflicts, that already existing occurrences remain valid is proven in the same way. 
    
    Therefore, let $G \Rightarrow_{r,m} H$ be a transformation that creates (at least one) such a new occurrence. 
    This means, there is an injective morphism $p^{\prime}: C \hookrightarrow H$ such that there does not exist an injective morphism $p_D: C \hookrightarrow G$ with $h \circ p_D = p^{\prime}$ (compare Figure~\ref{fig:interaction-rule-application-occurrence}, again), i.e., the application of $\crule{C}$ to $H$ at match $p$ is sequentially dependent on the transformation $G \Rightarrow_{r,m} H$. 
    We have to show that this occurrence is valid which means that there exists an injective morphism $q^{\prime}: C^{\prime} \hookrightarrow H$ such that $q^{\prime} \circ a = p^{\prime}$. 
    By the duality between conflicts and dependencies and the completeness result for weak critical pairs (see \cite[Remark~5.10]{EEPT06} and \cite[Lemma~6.4]{EGHLO12}) there is a weak critical sequence $G^{\prime} \Rightarrow_{r,m^{\prime}} A \Rightarrow_{\crule{C},p^{\prime\prime}} A$ that embeds into the sequence $G \Rightarrow_{r,m} H \Rightarrow_{\crule{C},p^{\prime}} H$ via a. 
    By assumption, there exists an injective morphism $q^{\prime\prime}: C^{\prime} \hookrightarrow A$ with $q^{\prime\prime} \circ a = p^{\prime\prime}$. 
    Moreover, $A$ embeds into $H$ via an injective morphism $k: A \hookrightarrow H$ such that $p^{\prime} = k \circ p^{\prime\prime}$ (by construction of weak critical pairs). 
    Hence, for $q^{\prime} \coloneqq k \circ q^{\prime\prime}$ we compute
    \begin{align*}
      q^{\prime} \circ a    & = k \circ q^{\prime\prime} \circ a \\
                        & = k \circ p^{\prime\prime} \\
                        & = p^{\prime}
    \end{align*}
    as desired. \qed
  \end{proof}
  
  \begin{proof}[of Theorem~\ref{thm:criterion-consistency-improvement}] 
    Here, we first consider the case of constraint $d = \forall (C, \exists C^{\prime})$.
    
    If $r$ is directly consistency improving, there exists a transformation step $G \Rightarrow_{r,m} H$ that constitutes a consistency improving rule application; in particular $G \not\models d$. 
    Hence, there exists an injective morphism $p: C \hookrightarrow G$ such that there exists no injective morphism $q: C^{\prime} \hookrightarrow G$ with $q^{\prime} \circ a = p$ (once more, compare Figure~\ref{fig:interaction-rule-application-occurrence} for the following).
    Moreover, either (i) for $p^{\prime} = h \circ p_D = \mathit{tr} \circ p$ (compare Lemma~\ref{lem:trace-morphism}) there exists an injective morphism $q^{\prime}: C^{\prime} \hookrightarrow H$ such that $q^{\prime} \circ a = p^{\prime}$, i.e., $q^{\prime} \models \exists C^{\prime}$, or (ii) there is no morphism $p_D: C \to D$ such that $p = g \circ p_D$.
    
    Assume (i). 
    There is a morphism $q^{\prime}: C^{\prime} \to H$ such that $p^{\prime} = q^{\prime} \circ a$. 
    The morphism $q^{\prime}$ is a match for the rule $\crule{C^{\prime}}$. 
    If $\crule{C^{\prime}}$ were sequentially independent from $r$, there was a morphism $q_D: C^{\prime} \to D$ such that $q^{\prime} = h \circ q_D$ where $D$ is the context object of the transformation step. 
    But, as in the proof of the above theorem, the resulting morphism $q \coloneqq g \circ q_D$ would satisfy $q \circ a = p$ which contradicts the assumption. 
    Hence, the application of $\crule{C^{\prime}}$ is sequentially dependent on the application of $r$. 
    
    Moreover, assume there to be an element
    \begin{equation*}
      x \in \left(n(R \setminus K) \cap p^{\prime}(C^{\prime})\right) \cap p^{\prime}(a(C)) \enspace .
    \end{equation*}
    By $x \in n(R \setminus K)$, $x \notin \mathit{tr} \circ p$ ($x$ has been newly created in $H$). 
    However, by $x \in p^{\prime}(a(C)$ and $p^{\prime} = \mathit{tr} \circ p$, $x \in \mathit{tr} \circ p$ which is a contradiction. 
    Hence, 
    \begin{equation*}
      x \in \left(n(R \setminus K) \cap p^{\prime}(C^{\prime})\right) \subseteq p^{\prime}(C^{\prime} \setminus a(C)) \enspace .
    \end{equation*}
    
    Assume (ii). 
    The morphism $p$ can be understood as match for $\crule{C}$ in $G$ and the non-existence of $p_D$ by definition means that the application of $r$ caused a conflict for the application of $\crule{C}$. 
    
    Secondly, we consider the case of constraint $c = \neg\exists C = \forall (C, \texttt{false})$. 
    Since no morphism $p$ can satisfy $\texttt{false}$, the first condition on direct consistency improvement can never be satisfied in that case. 
    Hence, to be directly consistency improving, the second condition must be true. 
    Again, as in (ii) above, the existence of a morphism $p: C \hookrightarrow G$ such that $\mathit{tr} \circ p$ is not total implies that the application of $r$ caused a conflict for $\crule{C}$. 
    Hence, directly consistency improving rules w.r.t. $c$ cause conflicts for $\crule{C}$.
    
    Conversely, assume $r$ to be a directly consistency sustaining rule w.r.t. $c$ such that $r$ causes a conflict for $\crule{C}$. 
    Hence, there is a transformation $t: G \Rightarrow_{r,m} H$ such that there exists an injective morphism $p: C \hookrightarrow G$ but no morphism $p_D: C \hookrightarrow D$ such that $p = g \circ p_D$. 
    By Lemma~\ref{lem:trace-morphism} (and since $p\not\models \texttt{false}$) this means that the second condition of the definition of direct consistency improvement is met. 
    Since additionally $r$ is directly consistency sustaining by assumption, this means that $r$ is directly consistency improving. \qed 
  \end{proof}
	
	}
\end{document}